\documentclass[
    aps,
    pra,
    reprint,
    amsmath,
    amssymb,
    superscriptaddress,
    nofootinbib
]{revtex4-2}

\usepackage{graphicx}
\usepackage{dcolumn}
\usepackage{bm}
\usepackage[utf8]{inputenc}
\usepackage[T1]{fontenc}
\usepackage{mathptmx}
\usepackage{hyperref}
\usepackage{amsmath}
\usepackage{amssymb}
\usepackage{amsthm}

\hypersetup{
    setpagesize = false,
    colorlinks  = true,
    urlcolor    = blue,
    linkcolor   = black,
    citecolor   = black,
    pdftitle    = {Toroidal-Core Certificates for PPT-Squared Diagonal Orthogonal Covariant Channels},
    pdfauthor   = {Samuel A. M\'arquez Gonz\'alez}
}

\newtheorem{theorem}{Theorem}
\newtheorem{lemma}{Lemma}
\newtheorem{proposition}{Proposition}
\newtheorem{corollary}{Corollary}
\newtheorem{remark}{Remark}

\newcommand{\EB}{\mathrm{EB}}
\newcommand{\TCP}{\mathrm{TCP}}

\begin{document}

\title{Toroidal-Core Certificates for PPT-Squared Diagonal Orthogonal Covariant Channels}

\author{Samuel A. M\'arquez Gonz\'alez}

\affiliation{
Department of Mathematical Sciences,
Rutgers University,
Camden, New Jersey 08102, USA
}

\email{sam959@scarletmail.rutgers.edu}

\begin{abstract}
The positive-partial-transpose-squared (PPT-squared) conjecture asks whether the composition of two positive-partial-transpose (PPT) quantum channels must be entanglement breaking (EB). The problem remains open in general and, within the diagonal orthogonal covariant (DOC) class, the first unresolved deterministic dimension is four. A sufficient condition for PPT-squared composition is derived for DOC maps in arbitrary finite dimension. The proof isolates the quantum coherence data into two correlation matrices, mixes each with the identity until it lies in the convex hull of rank-one correlation matrices, and converts the resulting objects into an explicit triplewise completely positive (TCP) core. The remaining contribution is purely classical and is TCP whenever a simple entrywise mixing inequality is satisfied. For PPT DOC maps $\Phi_{A,B,C}$ and $\Phi_{D,E,F}$ on $M_d$, the entrywise bound $(AD)_{ij}\ge 2\lfloor\sqrt{d}\rfloor\sqrt{A_{ii}A_{jj}D_{ii}D_{jj}}$ for all $i,j$ guarantees that $\Phi_{A,B,C}\circ\Phi_{D,E,F}$ is EB. In dimension four the universal coefficient becomes four. A continuous family of bistochastic PPT channels is then constructed; every member with $a\ne1$ is not EB, while every pairwise composition within the family is proved to be EB. A rank-sensitive refinement shows explicitly how low coherence rank can improve the universal coefficient. The resulting certificate is sufficient rather than necessary and is complementary to factor-width and semidefinite-hierarchy approaches.
\end{abstract}

\maketitle

\section{Introduction}

The positive-partial-transpose-squared (PPT-squared) conjecture asks whether composing two positive-partial-transpose (PPT) quantum channels must always produce an entanglement-breaking (EB) channel. The EB property is equivalent to separability of the Choi matrix, while PPT provides a natural larger channel class in which bound-entangled Choi states can occur \cite{HorodeckiShorRuskai2003,Ruskai2003QubitEB}. These channel-state correspondences rest on the Choi--Jamio\l kowski representation \cite{Jamiolkowski1972,Choi1975} and the PPT separability criterion \cite{Peres1996,Horodecki1996}. A broader account of the entanglement framework is given in Ref.~\cite{HorodeckiReview2009}.

Composition-to-EB questions arise in several forms. Entanglement-breaking indices and entanglement-saving channels quantify how repeated noise crosses the EB boundary \cite{LamiGiovannetti2015Indices,LamiGiovannetti2016Saving}. Eventual EB has also been studied through discrete and continuous dynamics \cite{RahamanJaquesPaulsen2018,HansonRouzeFranca2020,AhiableEtAl2021}, while channel divisibility provides a complementary dynamical viewpoint \cite{WolfCirac2008}. For PPT-squared itself, substantial progress is known in low dimensions and structured families \cite{KennedyManorPaulsen2018,ChristandlMullerHermesWolf2019,ChenYangTang2019,SinghNechita2022}. Recent work proves generic EB for compositions of random diagonal orthogonal covariant (DOC) channels, but leaves the deterministic DOC problem open \cite{NechitaPark2026}. Semidefinite-hierarchy experiments on structured local diagonal orthogonal invariant (LDOI) states likewise found no counterexample in the first unresolved deterministic dimension, $d=4$ \cite{BritzLaurent2025}.

Diagonal unitary and orthogonal symmetry is particularly useful because positivity, separability, and composition can be expressed directly through a small set of matrices \cite{SinghNechita2022,SinghNechita2021Quantum}. For DOC channels, an associated stochastic matrix controls the diagonal sector and also plays a central role in their ergodic behavior \cite{SinghDattaNechita2024}. The proof below combines this separation with a theorem of Kribs, Levick, Pereira, and Rahaman on complex correlation matrices \cite{KribsLevickPereiraRahaman2024}. If $C$ is a $d\times d$ correlation matrix and
\begin{equation}
m=\lfloor\sqrt{d}\rfloor,
\end{equation}
then
\begin{equation}
\frac{1}{m}C+\frac{m-1}{m}I
\label{eq:1}
\end{equation}
lies in the convex hull of rank-one complex correlation matrices. Such matrices will be called toroidal. This phase-mixture structure can be lifted directly into the triplewise completely positive (TCP) cone that characterizes EB in the DOC class \cite{SinghNechita2021Quantum}. Classical completely positive matrices provide the finite-dimensional geometric analog underlying the residual part of the construction \cite{Diananda1962,MaxfieldMinc1962,BermanShakedMonderer2003}.

The main outcome is a dimension-dependent sufficient condition involving only the diagonal-sector matrices $A$ and $D$. The coherence matrices $B,C,E,F$ determine two intermediate correlation matrices but disappear completely from the final inequality. The abstract criterion is then applied to an explicit continuous family $\Phi_{a,Z}$ of four-dimensional channels. Every member with $a\ne1$ is PPT and not EB, yet the composition of any two family members is EB. Thus the family is not merely an illustration of the criterion but a continuous set of nontrivial PPT-squared instances in the first unresolved deterministic DOC dimension. A short rank-sensitive refinement further quantifies how much the universal factor $\lfloor\sqrt{d}\rfloor$ can overestimate the mixing required when the coherence correlation matrices have low rank.

The result is deliberately a sufficient certificate, not a proposed characterization of PPT-squared DOC composition. The universal coefficient is saturated by a simple point of the family constructed below, but no claim is made that it marks an EB boundary. The objective is instead to isolate one analytically tractable mechanism by which classical mixing in the diagonal sector forces the quantum coherence sectors into a TCP decomposition.

From a physical viewpoint, the matrices $A$ and $D$ govern population transfer among the preferred basis sectors, while $B,C,E,F$ encode the associated coherence data. The criterion shows that sufficiently strong population mixing can dominate the scale set by those coherences: after the coherence contribution is absorbed into a toroidal TCP core, the remaining part is purely classical and entrywise nonnegative. In this sense, the certificate identifies a regime in which classical population mixing is strong enough to wash out the entanglement-carrying effect of the coherent sectors at the level of the composed channel.

\section{DOC maps and triplewise complete positivity}

Let $M_d$ denote the algebra of $d\times d$ complex matrices. For a vector $x\in\mathbb{C}^d$, $\operatorname{Diag}(x)$ denotes the diagonal matrix with diagonal $x$, while $\operatorname{diag}(X)$ is the diagonal vector of a matrix $X$. The Hadamard product is denoted by $\odot$, and $J_n$ denotes the $n\times n$ all-ones matrix.

A linear map $\Phi:M_d\to M_d$ is DOC when
\begin{equation}
\Phi(OXO)=O\Phi(X)O
\label{eq:2}
\end{equation}
for every diagonal orthogonal matrix $O$ with diagonal entries in $\{\pm1\}$. Such maps are parametrized by triples $(A,B,C)$ with a common diagonal \cite{SinghNechita2021Quantum}. In the convention used here,
\begin{equation}
\Phi_{A,B,C}(X)=\operatorname{Diag}(A\operatorname{diag}X)+\widetilde{B}\odot X+\widetilde{C}\odot X^{T},
\label{eq:3}
\end{equation}
where
\begin{equation}
\widetilde{B}=B-\operatorname{Diag}(\operatorname{diag}B),\qquad
\widetilde{C}=C-\operatorname{Diag}(\operatorname{diag}C),
\label{eq:4}
\end{equation}
and
\begin{equation}
\operatorname{diag}A=\operatorname{diag}B=\operatorname{diag}C.
\label{eq:5}
\end{equation}
Trace preservation is equivalent to every column of $A$ summing to one, while unitality is equivalent to every row of $A$ summing to one \cite{SinghNechita2021Quantum,SinghDattaNechita2024}.

The PPT condition is especially explicit. A DOC map is completely positive and completely copositive exactly when
\begin{equation}
A\ge0 \text{ entrywise},\qquad B\succeq0,\qquad C\succeq0,
\label{eq:6}
\end{equation}
and
\begin{equation}
A_{ij}A_{ji}\ge\max\{|B_{ij}|^2,|C_{ij}|^2\}
\label{eq:7}
\end{equation}
for every $i,j$ \cite{SinghNechita2022,SinghNechita2021Quantum}.

The EB condition is described by the TCP cone. A triple $(A,B,C)$ with common diagonal is TCP when there exist matrices $V,W\in M_{d\times r}$ such that
\begin{align}
A&=(V\odot\overline{V})(W\odot\overline{W})^{*},
\label{eq:8}\\
B&=(V\odot W)(V\odot W)^{*},
\label{eq:9}\\
C&=(V\odot\overline{W})(V\odot\overline{W})^{*}.
\label{eq:10}
\end{align}
The fundamental correspondence is
\begin{equation}
\Phi_{A,B,C}\in\EB\iff(A,B,C)\in\TCP_d.
\label{eq:11}
\end{equation}
This equivalence is established through the LDOI Choi representation \cite{SinghNechita2021Quantum}. It places the present problem in the same convex-geometric setting as classical complete positivity and pairwise complete positivity (PCP) \cite{Diananda1962,MaxfieldMinc1962,BermanShakedMonderer2003}.

The composition law is also closed in the DOC class. Let
\begin{equation}
\Phi=\Phi_{A,B,C},\qquad\Psi=\Phi_{D,E,F}.
\label{eq:12}
\end{equation}
Then
\begin{equation}
\Phi\circ\Psi=\Phi_{G,H,K}
\label{eq:13}
\end{equation}
with
\begin{equation}
G=AD,
\label{eq:14}
\end{equation}
\begin{equation}
H=B\odot E+C\odot F^{T}
+\operatorname{Diag}\Big(\operatorname{diag}\big(AD-B\odot E-C\odot F^{T}\big)\Big),
\label{eq:15}
\end{equation}
and
\begin{equation}
K=B\odot F+C\odot E^{T}
+\operatorname{Diag}\Big(\operatorname{diag}\big(AD-B\odot F-C\odot E^{T}\big)\Big).
\label{eq:16}
\end{equation}
Equivalent forms of these identities are used in the PPT-squared literature \cite{KennedyManorPaulsen2018,ChristandlMullerHermesWolf2019,SinghNechita2022}. In particular, for $i\ne j$,
\begin{equation}
H_{ij}=B_{ij}E_{ij}+C_{ij}F_{ji},
\label{eq:17}
\end{equation}
\begin{equation}
K_{ij}=B_{ij}F_{ij}+C_{ij}E_{ji},
\label{eq:18}
\end{equation}
while
\begin{equation}
H_{ii}=K_{ii}=G_{ii}.
\label{eq:19}
\end{equation}

\section{Toroidal TCP cores}

A correlation matrix is a positive semidefinite matrix with unit diagonal. A complex correlation matrix $R$ is toroidal when it belongs to the convex hull of rank-one correlation matrices,
\begin{equation}
R=\sum_{\mu}p_{\mu}u_{\mu}u_{\mu}^{*},\qquad
|u_{\mu i}|=1,\qquad p_{\mu}\ge0,\qquad\sum_{\mu}p_{\mu}=1.
\label{eq:20}
\end{equation}
Toroidal correlation matrices are precisely the correlation matrices associated with mixed-unitary Schur channels. Their role in mixed-unitary theory and their guaranteed neighborhood around the identity are developed in Ref.~\cite{KribsLevickPereiraRahaman2024}.

\begin{lemma}[Toroidal TCP core]
\label{lem:1}
Let $R,S\in M_d$ be toroidal correlation matrices and let $p\in\mathbb{R}_{+}^{d}$. Then
\begin{equation}
\big(pp^{T},\,D_pRD_p,\,D_pSD_p\big)\in\TCP_d,
\label{eq:21}
\end{equation}
where $D_p=\operatorname{Diag}(p)$.
\end{lemma}

\begin{proof}
Choose toroidal decompositions
\begin{equation}
R=\sum_{\mu}\alpha_{\mu}u_{\mu}u_{\mu}^{*},\qquad
S=\sum_{\nu}\beta_{\nu}v_{\nu}v_{\nu}^{*},
\label{eq:22}
\end{equation}
with $|u_{\mu i}|=|v_{\nu i}|=1$. Write
\begin{equation}
u_{\mu i}=e^{i\theta_{\mu i}},\qquad
v_{\nu i}=e^{i\varphi_{\nu i}},\qquad
\lambda_{\mu\nu}=\alpha_{\mu}\beta_{\nu}.
\label{eq:23}
\end{equation}
For every pair $(\mu,\nu)$, define vectors
\begin{align}
x^{(\mu\nu)}_i&=\lambda_{\mu\nu}^{1/4}\sqrt{p_i}\,
e^{i(\theta_{\mu i}+\varphi_{\nu i})/2},
\label{eq:24}\\
y^{(\mu\nu)}_i&=\lambda_{\mu\nu}^{1/4}\sqrt{p_i}\,
e^{i(\theta_{\mu i}-\varphi_{\nu i})/2}.
\label{eq:25}
\end{align}
Then
\begin{equation}
x^{(\mu\nu)}\odot\overline{x^{(\mu\nu)}}
=
y^{(\mu\nu)}\odot\overline{y^{(\mu\nu)}}
=
\sqrt{\lambda_{\mu\nu}}\,p,
\label{eq:26}
\end{equation}
because the phase factors cancel entrywise. Moreover,
\begin{equation}
x^{(\mu\nu)}\odot y^{(\mu\nu)}
=
\sqrt{\lambda_{\mu\nu}}\,D_pu_{\mu},
\label{eq:27}
\end{equation}
since
\[
e^{i(\theta_{\mu i}+\varphi_{\nu i})/2}
e^{i(\theta_{\mu i}-\varphi_{\nu i})/2}
=
e^{i\theta_{\mu i}},
\]
and
\begin{equation}
x^{(\mu\nu)}\odot\overline{y^{(\mu\nu)}}
=
\sqrt{\lambda_{\mu\nu}}\,D_pv_{\nu},
\label{eq:28}
\end{equation}
since
\[
e^{i(\theta_{\mu i}+\varphi_{\nu i})/2}
e^{-i(\theta_{\mu i}-\varphi_{\nu i})/2}
=
e^{i\varphi_{\nu i}}.
\]

The contribution of the pair $(\mu,\nu)$ to the first TCP component is therefore
\[
\big(x^{(\mu\nu)}\odot\overline{x^{(\mu\nu)}}\big)
\big(y^{(\mu\nu)}\odot\overline{y^{(\mu\nu)}}\big)^{*}
=
\lambda_{\mu\nu}pp^{T}.
\]
Summing over $\mu,\nu$ gives
\[
\sum_{\mu,\nu}\lambda_{\mu\nu}pp^{T}
=
\left(\sum_{\mu}\alpha_{\mu}\right)
\left(\sum_{\nu}\beta_{\nu}\right)pp^{T}
=
pp^{T}.
\]
For the second component,
\begin{align*}
\sum_{\mu,\nu}
\big(x^{(\mu\nu)}\odot y^{(\mu\nu)}\big)
\big(x^{(\mu\nu)}\odot y^{(\mu\nu)}\big)^{*}
&=
D_p
\left(
\sum_{\mu,\nu}\lambda_{\mu\nu}u_{\mu}u_{\mu}^{*}
\right)
D_p\\
&=
D_pRD_p,
\end{align*}
because $\sum_{\nu}\beta_{\nu}=1$. Similarly,
\begin{align*}
\sum_{\mu,\nu}
\big(x^{(\mu\nu)}\odot\overline{y^{(\mu\nu)}}\big)
\big(x^{(\mu\nu)}\odot\overline{y^{(\mu\nu)}}\big)^{*}
&=
D_p
\left(
\sum_{\mu,\nu}\lambda_{\mu\nu}v_{\nu}v_{\nu}^{*}
\right)
D_p\\
&=
D_pSD_p,
\end{align*}
because $\sum_{\mu}\alpha_{\mu}=1$. Hence the three TCP components are exactly
\[
\big(pp^{T},D_pRD_p,D_pSD_p\big),
\]
which proves Eq.~\eqref{eq:21}.
\end{proof}

A second elementary component will absorb the remaining diagonal-sector mass.

\begin{lemma}[Classical residual]
\label{lem:2}
If $L\in M_d(\mathbb{R})$ is entrywise nonnegative, then
\begin{equation}
\big(L,\operatorname{Diag}(\operatorname{diag}L),\operatorname{Diag}(\operatorname{diag}L)\big)\in\TCP_d.
\label{eq:29}
\end{equation}
\end{lemma}

\begin{proof}
For every $i\ne j$, the entry $L_{ij}$ is generated by the TCP atom
\begin{equation}
v=L_{ij}^{1/4}e_i,\qquad w=L_{ij}^{1/4}e_j.
\label{eq:30}
\end{equation}
Its first component is $L_{ij}E_{ij}$, while its second and third components vanish. For $i=j$, choose
\begin{equation}
v=w=L_{ii}^{1/4}e_i,
\label{eq:31}
\end{equation}
which contributes $L_{ii}E_{ii}$ to all three components. Summing the atoms proves Eq.~\eqref{eq:29}.
\end{proof}

\section{A sufficient PPT-squared criterion}

The toroidalization result needed below is a consequence of the rank bound for extreme complex correlation matrices. If $R$ is any $d\times d$ correlation matrix and
\begin{equation}
m=\lfloor\sqrt{d}\rfloor,
\label{eq:32}
\end{equation}
then
\begin{equation}
\widehat{R}=\frac{1}{m}R+\frac{m-1}{m}I
\label{eq:33}
\end{equation}
is toroidal \cite{KribsLevickPereiraRahaman2024}.

\begin{theorem}[Classical mixing criterion]
\label{thm:1}
Let $\Phi_{A,B,C}$ and $\Phi_{D,E,F}$ be PPT DOC maps on $M_d$, and let
\begin{equation}
m=\lfloor\sqrt{d}\rfloor.
\label{eq:34}
\end{equation}
If
\begin{equation}
(AD)_{ij}\ge2m\sqrt{A_{ii}A_{jj}D_{ii}D_{jj}}\quad\text{for all }i,j,
\label{eq:35}
\end{equation}
then
\begin{equation}
\Phi_{A,B,C}\circ\Phi_{D,E,F}\in\EB.
\label{eq:36}
\end{equation}
\end{theorem}

\begin{proof}
Let $(G,H,K)$ be the DOC triple of the composition. By Eq.~\eqref{eq:14},
\begin{equation}
G=AD.
\label{eq:37}
\end{equation}
Define
\begin{equation}
S_H=B\odot E+C\odot F^{T},\qquad
S_K=B\odot F+C\odot E^{T}.
\label{eq:38}
\end{equation}
Because $B,C,E,F$ are Hermitian positive semidefinite, their transposes satisfy $E^{T}=\overline{E}\succeq0$ and $F^{T}=\overline{F}\succeq0$. The Schur product theorem therefore gives
\begin{equation}
S_H\succeq0,\qquad S_K\succeq0.
\label{eq:39}
\end{equation}
The common-diagonal property of the DOC triples yields
\begin{equation}
(S_H)_{ii}=(S_K)_{ii}=2A_{ii}D_{ii}.
\label{eq:40}
\end{equation}
Set
\begin{equation}
q_i=\sqrt{2A_{ii}D_{ii}}.
\label{eq:41}
\end{equation}
On the support of $q$, define
\begin{equation}
R_{ij}=\frac{(S_H)_{ij}}{q_iq_j},\qquad
S_{ij}=\frac{(S_K)_{ij}}{q_iq_j}.
\label{eq:42}
\end{equation}
If $q_i=0$, positive semidefiniteness forces the corresponding row and column of both $S_H$ and $S_K$ to vanish. Hence $R$ and $S$ can be extended by an identity block on the zero support. The resulting matrices are correlation matrices.

Define their toroidal mixtures
\begin{equation}
\widehat{R}=\frac{1}{m}R+\frac{m-1}{m}I,\qquad
\widehat{S}=\frac{1}{m}S+\frac{m-1}{m}I.
\label{eq:43}
\end{equation}
Both are toroidal by Eq.~\eqref{eq:33}. Let
\begin{equation}
p_i=\sqrt{m}\,q_i=\sqrt{2mA_{ii}D_{ii}}.
\label{eq:44}
\end{equation}
Lemma~\ref{lem:1} implies
\begin{equation}
T_0=\big(pp^{T},\,D_p\widehat{R}D_p,\,D_p\widehat{S}D_p\big)\in\TCP_d.
\label{eq:45}
\end{equation}
For $i\ne j$,
\begin{equation}
(D_p\widehat{R}D_p)_{ij}=mq_iq_j\frac{R_{ij}}{m}=(S_H)_{ij}=H_{ij},
\label{eq:46}
\end{equation}
\begin{equation}
(D_p\widehat{S}D_p)_{ij}=mq_iq_j\frac{S_{ij}}{m}=(S_K)_{ij}=K_{ij}.
\label{eq:47}
\end{equation}
Thus the toroidal core reproduces every off-diagonal coherence of the composed triple.

Its first component satisfies
\begin{equation}
(pp^{T})_{ij}=2m\sqrt{A_{ii}A_{jj}D_{ii}D_{jj}}.
\label{eq:48}
\end{equation}
Define
\begin{equation}
L=G-pp^{T}.
\label{eq:49}
\end{equation}
The hypothesis \eqref{eq:35} is exactly the assertion that $L$ is entrywise nonnegative. Lemma~\ref{lem:2} therefore gives
\begin{equation}
T_1=\big(L,\operatorname{Diag}(\operatorname{diag}L),\operatorname{Diag}(\operatorname{diag}L)\big)\in\TCP_d.
\label{eq:50}
\end{equation}
Equations~\eqref{eq:46}, \eqref{eq:47}, and the common-diagonal identity \eqref{eq:19} imply the exact decomposition
\begin{equation}
(G,H,K)=T_0+T_1.
\label{eq:51}
\end{equation}
The TCP cone is convex, so $(G,H,K)\in\TCP_d$. Equation~\eqref{eq:11} then gives the entanglement-breaking conclusion.
\end{proof}

The theorem is notable for the absence of $B,C,E,F$ from Eq.~\eqref{eq:35}. The coherence data determine the two correlation matrices in Eq.~\eqref{eq:42}, but toroidalization absorbs them without a phase-by-phase estimate. The only remaining requirement concerns how strongly the diagonal-sector matrix product $AD$ mixes the basis populations.

For $d=4$, the dimension factor becomes $m=2$.

\begin{corollary}[Dimension four]
\label{cor:1}
Let $\Phi_{A,B,C}$ and $\Phi_{D,E,F}$ be PPT DOC maps on $M_4$. If
\begin{equation}
(AD)_{ij}\ge4\sqrt{A_{ii}A_{jj}D_{ii}D_{jj}}\quad\text{for all }i,j,
\label{eq:52}
\end{equation}
then their composition is EB.
\end{corollary}

For an auto-composition with the same diagonal-sector matrix $A$, Eq.~\eqref{eq:52} becomes
\begin{equation}
(A^2)_{ij}\ge4A_{ii}A_{jj}.
\label{eq:53}
\end{equation}
The criterion is only sufficient. No claim of necessity is made. In particular, the coefficient four is saturated by the family of Sec.~V at $a=b=1$, where $M_1=J_4$ and $M_1M_1=4J_4$. This is saturation of the present certificate, not evidence that Eq.~\eqref{eq:52} is a necessary entanglement-breaking boundary.

\begin{corollary}[Rank-sensitive refinement]
\label{cor:2}
Let $R$ and $S$ be the correlation matrices defined on the nonzero support of $q$ in Eq.~\eqref{eq:42}, and set
\begin{equation}
\kappa=\min\big\{\lfloor\sqrt{d}\rfloor,\,\max\{\operatorname{rank}R,\operatorname{rank}S\}\big\}.
\label{eq:54}
\end{equation}
If
\begin{equation}
(AD)_{ij}\ge2\kappa\sqrt{A_{ii}A_{jj}D_{ii}D_{jj}}\quad\text{for all }i,j,
\label{eq:55}
\end{equation}
then $\Phi_{A,B,C}\circ\Phi_{D,E,F}$ is EB.
\end{corollary}

\begin{proof}
Let $I_q=\{i:q_i>0\}$ be the nonzero support of $q$. If $I_q=\varnothing$, then $S_H=S_K=0$ because both matrices are positive semidefinite with zero diagonal. Hence $H=K=\operatorname{Diag}(\operatorname{diag}G)$, and Lemma~\ref{lem:2}, applied directly to the entrywise nonnegative matrix $G=AD$, proves that the composed triple is TCP.

Assume now that $I_q\ne\varnothing$, and let $n=|I_q|$. For a correlation matrix $T$ of rank $r$, Ref.~\cite{KribsLevickPereiraRahaman2024} shows that
\[
\frac{1}{r}T+\frac{r-1}{r}I
\]
is toroidal. On the active support, the same reference gives the universal denominator $\lfloor\sqrt{n}\rfloor$, and $\lfloor\sqrt{n}\rfloor\le m=\lfloor\sqrt{d}\rfloor$. If a toroidal mixture is available with denominator $k_0$, then every larger denominator $k\ge k_0$ is also valid because the new mixture is a convex combination of the old mixture and $I$, which is toroidal. Thus a valid denominator for each active-support correlation matrix $T\in\{R,S\}$ is $\min\{\lfloor\sqrt{n}\rfloor,\operatorname{rank}T\}$, and the single denominator $\kappa$ in Eq.~\eqref{eq:54} is at least this value for both $R$ and $S$. Therefore $\kappa$ works simultaneously for the two matrices on $I_q$.

Choose toroidal decompositions of the corresponding $\kappa$-mixtures on $I_q$. Each unimodular phase vector in those decompositions may be extended arbitrarily to the complementary coordinates while preserving unit modulus. These extensions produce toroidal correlation matrices on all of $M_d$ whose principal submatrices on $I_q$ are the desired mixtures. Since $p_i=\sqrt{\kappa}\,q_i$ vanishes outside $I_q$, the matrices $D_p\widehat{R}D_p$ and $D_p\widehat{S}D_p$ depend only on those active principal submatrices. Lemma~\ref{lem:1} is therefore applicable in dimension $d$ exactly as in Theorem~\ref{thm:1}. Repeating that proof with $m$ replaced by $\kappa$ proves the claim.
\end{proof}

\begin{remark}[A strict gain over the universal coefficient]
\label{rem:1}
The rank refinement can certify examples missed by the universal $d=4$ coefficient. Let
\begin{equation}
A=D=
\begin{pmatrix}
2/5 & 3/5 & 0 & 0\\
3/5 & 2/5 & 0 & 0\\
0 & 0 & 0 & 1\\
0 & 0 & 1 & 0
\end{pmatrix}
\label{eq:56}
\end{equation}
and
\begin{equation}
B=C=E=F=\frac{2}{5}
\begin{pmatrix}
1 & 1 & 0 & 0\\
1 & 1 & 0 & 0\\
0 & 0 & 0 & 0\\
0 & 0 & 0 & 0
\end{pmatrix}.
\label{eq:57}
\end{equation}
The corresponding DOC maps are trace preserving, unital, and PPT by Eqs.~\eqref{eq:6}--\eqref{eq:7}. On the nonzero support of $q$, Eq.~\eqref{eq:42} gives $R=S=J_2$, so $\operatorname{rank}R=\operatorname{rank}S=1$ and $\kappa=1$. Meanwhile,
\begin{equation}
(A^2)_{11}=\frac{13}{25},\qquad(A^2)_{12}=\frac{12}{25}.
\label{eq:58}
\end{equation}
The universal four-dimensional test would require these active entries to be at least $4(2/5)^2=16/25$, and therefore fails. Corollary~\ref{cor:2} requires only $2(2/5)^2=8/25$, which is satisfied. Thus the rank information produces a genuine enlargement of the certified region rather than a merely formal change of constants.
\end{remark}

\section{A continuous family in dimension four}

The criterion can be applied to channels that are genuinely PPT but not EB. Let $a>0$ and define
\begin{equation}
r_a=\frac{a+a^{-1}}{2},\qquad s_a=1+3r_a,
\label{eq:59}
\end{equation}
with
\begin{equation}
M_a=
\begin{pmatrix}
1 & a & a^{-1} & r_a\\
a^{-1} & 1 & a & r_a\\
a & a^{-1} & 1 & r_a\\
r_a & r_a & r_a & 1
\end{pmatrix}.
\label{eq:60}
\end{equation}
Every row and every column of $M_a$ sums to $s_a$. Let $Z\in M_4$ be an arbitrary correlation matrix. In particular, the coherence sector $C_{a,Z}$ is not fixed to the all-ones choice and may vary over the full correlation-matrix set. Define
\begin{equation}
\Phi_{a,Z}=\Phi_{A_a,B_a,C_{a,Z}},
\label{eq:61}
\end{equation}
where
\begin{equation}
A_a=\frac{M_a}{s_a},\qquad B_a=\frac{J_4}{s_a},\qquad C_{a,Z}=\frac{Z}{s_a}.
\label{eq:62}
\end{equation}

\begin{proposition}
\label{prop:1}
For every $a>0$ and every correlation matrix $Z$, the map $\Phi_{a,Z}$ is a trace-preserving unital PPT channel. If $a\ne1$, it is not EB.
\end{proposition}

\begin{proof}
The row and column sums of $A_a$ are one, so the map is unital and trace preserving. Both $B_a$ and $C_{a,Z}$ are positive semidefinite. For indices $i,j\in\{1,2,3\}$ with $i\ne j$,
\begin{equation}
(M_a)_{ij}(M_a)_{ji}=1.
\label{eq:63}
\end{equation}
For a pair containing index four,
\begin{equation}
(M_a)_{i4}(M_a)_{4i}=r_a^2\ge1.
\label{eq:64}
\end{equation}
Since every correlation matrix satisfies $|Z_{ij}|\le1$,
\begin{equation}
(A_a)_{ij}(A_a)_{ji}\ge\max\{|(B_a)_{ij}|^2,|(C_{a,Z})_{ij}|^2\}.
\label{eq:65}
\end{equation}
The PPT characterization \eqref{eq:6}--\eqref{eq:7} proves the first statement.

It remains to exclude EB when $a\ne1$. If the full triple in Eq.~\eqref{eq:62} were TCP, then the pair $(A_a,B_a)$ would be PCP, and the same would remain true after restriction to the principal indices $1,2,3$ \cite{SinghNechita2021Quantum}. Up to the positive factor $s_a^{-1}$, the restricted pair is
\begin{equation}
\widetilde{X}_a=
\begin{pmatrix}
1 & a & a^{-1}\\
a^{-1} & 1 & a\\
a & a^{-1} & 1
\end{pmatrix},\qquad J_3.
\label{eq:66}
\end{equation}
Britz and Laurent studied the corresponding conjugate local diagonal unitary invariant (CLDUI) state $\rho_{a,a'}$ and proved in Theorem~3.21 of Ref.~\cite{BritzLaurent2025} that it is separable exactly when $a\ge1$ and $a'\ge1$. Here $a'=a^{-1}$, so both conditions hold simultaneously only for $a=1$. Therefore $a\ne1$ implies that the restricted pair is not PCP, hence the full triple is not TCP. Equation~\eqref{eq:11} shows that $\Phi_{a,Z}\notin\EB$.
\end{proof}

The next lemma supplies the classical mixing inequality needed by Corollary~\ref{cor:1}.

\begin{lemma}
\label{lem:3}
For every $a,b>0$,
\begin{equation}
M_aM_b\ge4J_4
\label{eq:67}
\end{equation}
entrywise.
\end{lemma}

\begin{proof}
Inside the leading $3\times3$ block, the possible entry types are
\begin{equation}
1+\frac{a}{b}+\frac{b}{a}+r_ar_b,
\label{eq:68}
\end{equation}
\begin{equation}
a+b+\frac{1}{ab}+r_ar_b,
\label{eq:69}
\end{equation}
and
\begin{equation}
ab+\frac{1}{a}+\frac{1}{b}+r_ar_b.
\label{eq:70}
\end{equation}
The first is at least four because $a/b+b/a\ge2$ and $r_ar_b\ge1$. In Eq.~\eqref{eq:69}, the first three displayed positive terms have product one, so their sum is at least three by AM--GM, and again $r_ar_b\ge1$. The same argument applies to Eq.~\eqref{eq:70}.

For an entry in the fourth column with row among $1,2,3$,
\begin{equation}
(M_aM_b)_{i4}=r_b(1+a+a^{-1})+r_a\ge4.
\label{eq:71}
\end{equation}
The fourth row is analogous, and
\begin{equation}
(M_aM_b)_{44}=3r_ar_b+1\ge4.
\label{eq:72}
\end{equation}
Thus every entry of $M_aM_b$ is at least four.
\end{proof}

\begin{theorem}[Pairwise composition of the family]
\label{thm:2}
For all $a,b>0$ and for arbitrary correlation matrices $Z,W\in M_4$,
\begin{equation}
\Phi_{a,Z}\circ\Phi_{b,W}\in\EB.
\label{eq:73}
\end{equation}
In particular, if $a\ne1$ and $b\ne1$, both factors are PPT and non-EB while their composition is EB.
\end{theorem}

\begin{proof}
For the two diagonal-sector matrices,
\begin{equation}
A_aA_b=\frac{M_aM_b}{s_as_b}.
\label{eq:74}
\end{equation}
Every diagonal entry of $A_a$ equals $s_a^{-1}$, and every diagonal entry of $A_b$ equals $s_b^{-1}$. The right-hand side of the four-dimensional criterion \eqref{eq:52} is therefore
\begin{equation}
4\sqrt{(A_a)_{ii}(A_a)_{jj}(A_b)_{ii}(A_b)_{jj}}=\frac{4}{s_as_b}.
\label{eq:75}
\end{equation}
Lemma~\ref{lem:3} gives
\begin{equation}
(A_aA_b)_{ij}\ge\frac{4}{s_as_b}
\label{eq:76}
\end{equation}
for every $i,j$. Corollary~\ref{cor:1} proves Eq.~\eqref{eq:73}. Proposition~\ref{prop:1} proves the final statement.
\end{proof}

The family contains an exact analytic version of the type of PPT-entangled LDOI inputs used in recent semidefinite-hierarchy tests of PPT-squared behavior \cite{BritzLaurent2025}. The matrices differ in the fourth row and column because Eq.~\eqref{eq:60} is chosen so that the resulting DOC maps are simultaneously trace preserving and unital. The upper-left $3\times3$ obstruction to separability is retained, while the choice $r_a=(a+a^{-1})/2$ supplies enough classical mixing to satisfy Eq.~\eqref{eq:52} for every pair of parameters. Importantly, this conclusion is uniform over the arbitrary correlation matrices $Z$ and $W$ entering the two coherence sectors, so the construction is not tied to the special choice $Z=W=J_4$.

\section{Relation to other certificates}

The result is complementary to existing approaches to PPT-squared composition. Choi-type and diagonal-covariant cases have been attacked through explicit cone factorizations and factor-width arguments \cite{SinghNechita2022,NechitaPark2026}. More specifically, Proposition~6.2 of Ref.~\cite{NechitaPark2026} gives deterministic sufficient conditions for DOC compositions that retain explicit pairwise coherence combinations involving $B,C,E,F$. By contrast, the condition in Eq.~\eqref{eq:35} depends only on the diagonal-sector matrices $A$ and $D$. No general dominance relation between the two certificates is claimed. Semidefinite hierarchies provide another route to separability certification and can exploit the LDOI block structure efficiently \cite{BritzLaurent2025}. The general Doherty--Parrilo--Spedalieri (DPS) framework originates from symmetric-extension criteria for separability and remains one of the standard systematic approximations to the separable cone.

The toroidal-core proof uses a different decomposition. Rather than estimating each pairwise coherence separately, Eqs.~\eqref{eq:38}--\eqref{eq:43} package all off-diagonal data into two correlation matrices. The universal mixing result of Ref.~\cite{KribsLevickPereiraRahaman2024} then turns those correlation matrices into rank-one phase mixtures, while Lemma~\ref{lem:1} lifts the two mixtures simultaneously into a single TCP object. The residual in Lemma~\ref{lem:2} carries no off-diagonal quantum coherence and is controlled entirely by the entrywise inequality \eqref{eq:35}.

This mechanism explains both the strength and the limitation of the criterion. It is insensitive to the detailed phases of $B,C,E,F$, which makes it robust over continuous channel families. The price is that the universal choice $m=\lfloor\sqrt{d}\rfloor$ is a worst-case toroidalization guarantee, not an optimized parameter for a given composition. Corollary~\ref{cor:2} and Remark~\ref{rem:1} make this loss explicit. In the displayed four-dimensional example the universal coefficient is four, while the actual coherence ranks reduce it to two; the universal test fails and the rank-sensitive test succeeds. This also clarifies the status of tightness. Theorem~\ref{thm:1} is a sufficient condition, and no necessity claim or general optimality claim for its coefficient is made.

The classical residual also suggests possible refinements when additional symmetry or positivity is available. In dimensions at most four, doubly nonnegative and completely positive matrices coincide, whereas the cones separate in higher dimension \cite{Diananda1962,MaxfieldMinc1962,BermanShakedMonderer2003}. The present proof does not require the residual $L$ to be symmetric or positive semidefinite; entrywise nonnegativity alone is sufficient for Lemma~\ref{lem:2}. Nevertheless, the low-dimensional coincidence may be useful for stronger four-dimensional DOC (DOC4) certificates in subclasses where the residual has additional symmetric positive-semidefinite structure. Such extensions are left outside the present scope.

Several neighboring literatures address different forms of entanglement degradation. One-shot entanglement distribution under local noise has been studied through singlet fraction and related optimization questions \cite{Pal2014Singlet,Streltsov2015Unified,SiddhuSmolin2023Optimal}. Source placement is a distinct architectural problem \cite{MasajadaFellousStreltsov2026}. Resource-oriented analyses include entanglement cost, distribution by separable states, discord-assisted distribution, and excessive distribution \cite{Streltsov2012QuantumCost,Cubitt2003Separable,Chuan2012Discord,Zuppardo2016Excessive}. Entanglement-annihilating maps provide another nearby notion \cite{MoravcikovaZiman2010,FilippovRybarZiman2012,FilippovZiman2013}, with further structural developments in Refs.~\cite{LamiHuber2016,AubrunMullerHermes2023}. Recent generalizations based on Schmidt-number degradation and ordered-cone methods extend this hierarchy \cite{MallickGangulyMajumdar2026,LaPianaMullerHermes2026}. Negativity remains a useful diagnostic in numerical studies even though the present arguments use exact cone membership \cite{VidalWerner2002}.

The broader channel literature also contains methods for transferring nonunitary evolution to better structured representatives. Canonical descriptions of quantum channels and unital maps are developed in Refs.~\cite{FujiwaraAlgoet1999,RuskaiSzarekWerner2002,BraunEtAl2014,ChoiLi2023}. Sinkhorn-type scaling and invertible local filtering provide another route to normal forms \cite{Sinkhorn1964,Gurvits2004,GeorgiouPavon2015,Cariello2019}. Their use in quantum-channel and entanglement problems is illustrated in Refs.~\cite{Filippov2021Sinkhorn,VerstraeteDehaeneDeMoor2001,VerstraeteDehaeneDeMoor2003,FilippovFrizenKolobova2018}. The theorem above occupies a different position. It stays inside the DOC parametrization and supplies a one-step deterministic certificate for two possibly different PPT maps by reading sufficient mixing directly from their diagonal sectors.

\section{Conclusion}

A sufficient PPT-squared criterion has been established for DOC maps in arbitrary finite dimension. The proof converts the two coherence sectors of a PPT composition into normalized correlation matrices, applies a universal toroidal mixing theorem, and lifts the resulting phase mixtures into a triplewise completely positive core. Once that core is removed, the remaining contribution is TCP whenever the classical product $AD$ dominates a rank-one nonnegative matrix determined by the diagonal entries of $A$ and $D$.

For $d=4$, the criterion reduces to
\begin{equation}
(AD)_{ij}\ge4\sqrt{A_{ii}A_{jj}D_{ii}D_{jj}}.
\label{eq:77}
\end{equation}
A continuous family of bistochastic PPT channels was constructed for which this inequality holds pairwise. For every $a\ne1$, the corresponding channel is not EB, yet the composition of any two family members is EB. The example gives an analytic family of nontrivial PPT-squared instances in the first dimension where the deterministic DOC problem remains open.

The result does not settle the full DOC4 conjecture and is not intended as a replacement for factor-width or semidefinite-hierarchy methods. Its contribution is a distinct structural certificate in which the quantum coherence information is absorbed globally by toroidal correlation matrices and the final hypothesis is purely classical. Physically, this separation suggests that sufficiently strong mixing of basis populations can overwhelm the entanglement-carrying role of the coherent sectors, even when each individual PPT channel is not EB. This viewpoint may be useful in identifying further deterministic subclasses of PPT channels whose compositions necessarily become EB.

\bibliographystyle{apsrev4-2}
\bibliography{referencias}

\end{document}